\documentclass[conference]{IEEEtran}
\usepackage{graphicx,epsfig,color,graphics}
\usepackage{amssymb}
\usepackage{amsmath}
\newtheorem{theorem}{Theorem}
\newtheorem{lemma}{Lemma}

\begin{document}

\title{Lower Bound for the Communication Complexity of the Russian Cards Problem}

%\title{A Combinatorial High Girth Tanner Graph Construction for a 
%       Family of Near-Regular LDPC Codes}


% AUTHOR BEGINS HERE 

\author{\authorblockN{Aiswarya Cyriac, K. Murali Krishnan, }
  \authorblockA{Department of Computer Science and Engineering\\
                National Institute of Technology, Calicut 673601,  India.\\
		Email: aiswaryanitc@gmail.com, kmuralinitc@gmail.com}}

%\author{\authorblockN{
%\authorblockA{School of Electrical and\\Computer Engineering\\
%Georgia Institute of Technology\\
%Atlanta, Georgia 30332--0250\\
%Email: mshell@ece.gatech.edu}
%\and
%\authorblockA{Twentieth Century Fox\\
%\and
%\authorblockN{James Kirk\\ and Montgomery Scott}
%\authorblockA{Starfleet Academy\\
%San Francisco, California 96678-2391\\


%Homer Simpson\authorrefmark{2},
%Montgomery Scott\authorrefmark{3} and
%Eldon Tyrell\authorrefmark{4}}
%Telephone: (800) 555--1212, Fax: (888) 555--1212}
%\authorblockA{\authorrefmark{4}Tyrell Inc., 123 Replicant Street, Los Angeles, California 90210--4321}}

\maketitle

\begin{abstract}
In this paper it is shown that no public announcement scheme that can be modeled in Dynamic Epistemic Logic (DEL) can solve the Russian Cards Problem (RCP) in one announcement. Since DEL is a general model for any public announcement scheme \cite{dit4,baltag,benthem1lonely,lutz,dit5}  we conclude that there exist no single announcement solution to the RCP. The proof demonstrates the utility of DEL in proving lower bounds for communication protocols. It is also shown that a general version of RCP has no two announcement solution when the adversary has sufficiently large number of cards.
\vspace{0.5cm}

\textbf{Key words:} Russian cards problem, Dynamic Epistemic Logic, Communication complexity, Lower bound.
%Performance simulations of iterative decoding algorithm on 
%codes designed using the method are reported.
\end{abstract}

\IEEEpeerreviewmaketitle

\section{Introduction}
\label{intro}

In the Russian cards problem  (RCP), there are three players and seven cards. The cards are randomly distributed among the players such that two players get three cards each and the third player gets one card. The problem is to find a sequence of public announcements by which players with three cards each are able to acquire complete information about all the cards,  without the third player knowing about any of their cards.

The solution to the problem will imply a method to communicate information among parties in a distributed computing setting securely without using any encryption \cite{fischer,jam,koizumi}. The analogy is that, the communicating agents and adversaries are modeled as players and the information to be communicated as the ownership of cards. It is generally believed that the above game gives unconditional security \cite{dit4,fischer,jam,koizumi}.

 Various solutions to the above problem can be found in the literature \cite{dit-sol1,dit-sol2,makar-sol1}. They all require at least two public announcements. Here we address the problem of formally establishing that no public announcement scheme can solve the problem in fewer announcements.  The framework of
Dynamic Epistemic Logic (DEL) is used to establish the lower bound. Similar bounds using other models for related problems can be found in \cite{albert,ambainis,cai,fischer,franklin}.

The following sections briefly discuss dynamic epistemic logic,  modeling of RCP in DEL and a proof for the lower bound. Finally a generalization of the RCP is presented and it is shown that two announcements are not sufficient to solve the general case when the adversary has sufficiently large number of cards.

\section{Russian Cards Problem (RCP)}

The problem was posed in 2000 \cite{dit4} as the following: 

\hspace*{.5 cm}\textit{From a pack of seven known cards two players each draw three cards and the third player gets the remaining card. How can the players with three cards openly(publicly) inform each other about their cards, without the third player learning from any of their cards who holds it?}

Let us call the cards 0, 1, ..., 6. The players are Anne, Bill and Cath. Anne and Bill have three cards each and Cath has one. No secret communication is possible. Only announcements are allowed. The announcements are assumed to be truthful and public. Through a sequence of such announcements Anne and Bill have to learn the actual deal of the cards. i.e., for each card from the above pack, Anne and Bill should be able to say to whom that card belongs. Also for any card from the pack other than the one Cath is holding, she should not be able to say to whom that card belongs.

Various solutions to the above problem can be seen in the literature \cite{dit-sol1,dit-sol2,makar-sol1}. All these solutions require two announcements.

\section{Dynamic Epistemic Logic ( DEL)}

 In this paper we express the RCP in the framework of Dynamic Epistemic Logic (DEL). This section briefly presents the syntax and semantics of Dynamic Epistemic Logic. More detailed discussion of the DEL, examples and its applications can be found in \cite{baltag,benthem,dru,gerbrandy,gerbrandy2}.

 \textit{Kripke models} and \textit{action models} are semantical structures of dynamic epistemic logic. Given a set of agents (players) and basic propositions  a Kripke model consists of the set of possible states and \textit{accessibility}  (or \textit{indistinguishability}) relation between the states for every agent. The knowledge of the players about the state of the game in imperfect information games\footnote{In imperfect information games, players do not have complete information about other players' moves. } can be modeled by viewing game states as Kripke states and players as agents.

 Action models model the actions of the players that will alter the knowledge of the players. Given an initial Kripke model modeling the knowledge of the players, the action models can be sequentially executed in the Kripke model, resulting in a new Kripke model that models the knowledge of the players after the action .

\paragraph*{Epistemic Language}
Epistemic logic can be used to model knowledge in games of imperfect information \cite{baltag2,benthem,benthem1lonely,benthem3,hintika}.

 Let $N$ be a finite set of agents\footnote{Players are modeled as agents. Agents are assumed to be perfect logicians, i.e. the agents know all the consequences of their knowledge.} and $P$ be a finite set of propositional atoms. The Epistemic language $L_{P,N}$ is the smallest closed set for which the following holds:
\begin{itemize}
 \item $p \in P \Rightarrow p \in L_{P,N} $
\item $\phi,\psi \in L_{P,N} \Rightarrow \neg \phi, (\phi \wedge \psi)  \in L_{P,N} $
\item $\phi \in L_{P,N}$ and $n \in N \Rightarrow K_{n}\phi \in L_{P,N}$
\end{itemize}

The sentence $K_{1}\phi$ is read as: agent 1 knows $\phi$ .  $(\phi \vee \psi), \phi \rightarrow \psi$, and $ \phi \leftrightarrow \psi  $ and are abbreviations for $\neg (\neg \phi \wedge \neg \psi ) , \neg \phi \vee \psi$,  and $ ( \phi \wedge \psi) \vee ( \neg \phi \wedge \neg \psi) $ respectively. The notation $\top$ stands for $\neg ( p \wedge \neg p)$ for some $p \in P$.\\

Let a finite set of agents $N$ and a finite set of propositional atoms $P$ be given. A Kripke model \cite{voorbrak,benthem,baltag2,dru} is a tuple $(W,R,V)$ where:
\begin{itemize}
 \item The set $W$ is nonempty set of states $\{w_{1},\ldots w_{|W|}\}$
 \item The accessibility function $R: N \rightarrow 2^{W\times W}$  assigns for each agent $n \in N $ a set  of ordered pairs of states. $\forall n \in N$, $ R(n) \subseteq W\times W$ is an equivalence relation.
\item The valuation function  $V:W\rightarrow2^{P}$ assign to each state a set of propositional atoms. $\forall w \in W$, $ V(w) \subseteq P $.\\
\end{itemize}

$(w,w')\in R(n)$ is interpreted as  state $w'$ is accessible (or indistinguishable) from state $w$ for the agent $n$.
The set of propositional atoms assigned to a state by $V$ is the atoms which hold in that state. 

A Kripke world is a pair  consisting of a Kripke model $M=(W,R,V)$ and a state $w\in W$ and is denoted by $(M,w)$.

\label{del}
 \paragraph*{Semantics of Epistemic Logic:} Let a Kripke model $M=(W,R,V)$  and the epistemic language $L_{P,N}$ be given.
\begin{itemize}
 \item $  M,w \models p \Leftrightarrow p \in V(w)$
\item $  M,w \models \neg \phi \Leftrightarrow M,w \not\models \phi $
\item $  M,w \models  \phi \wedge \phi \Leftrightarrow M,w \models \phi $ and $ M,w \models \psi $
\item $  M,w \models K_{n} \phi \Leftrightarrow $ For all $w'$ such that $(w,w') \in R(n)$, $M,w' \models \phi$
\end{itemize}

A Player $x$ knows the fact $\phi$ in the state $w$ only if $\phi$ holds in all the states indistinguishable from $w$. Also if $\phi$ is true in all the states indistinguishable from $w$ for Player $x$, she can infer  $\phi$.

 The \textit{action models} are used to update Kripke Models. An action model consist of a set of actions, an accessibility       (indistinguishability) relation between the actions for every agent, and a precondition function for each action.

 Let a set of agents $N$ and the epistemic language $L_{P,N}$ be given. An Action model $\mu$ is a tuple $(A,R^{*},\Pi)$:
\begin{itemize}
 \item The set $A$ is the nonempty set of actions $\{a_{1},\ldots ,a_{|A|}\}$
\item The accessibility function $R^{*}: N\rightarrow 2^{A \times A}$ is a function which assigns to each agent a set of ordered pairs of actions. $\forall n\in N$, $R^{*}(n) \subseteq A\times A$ is an equivalence relation. 
\item The precondition function $\Pi : A \rightarrow L_{P,N}$ assigns to every action a precondition. $\forall a \in A$, $ \Pi(a) \in L_{P,N}$
\end{itemize}
\paragraph*{Execution} Action models are used to update Kripke model. Thus an action model is an operator on a Kripke Model. The execution of an action in a state results in a new state if and only if the precondition of the action holds in that state. 

\label{action}
Let a Kripke model $M=(W,R,V)$ and an action model $\mu = (A,R^{*},\Pi)$ be given. The execution of action model $\mu$ in Kripke model $M$ results in a Kripke model denoted by $M \otimes \mu$. $M \otimes \mu = (W',R',V') $ such that:
\begin{itemize}
\item The set of worlds  $W' = \{ ( w,a ) \in W \times A$ $ | $ $M,w \models \Pi( a)\}$
\item For every Player $n \in N$, $((w,a),(w',a')) \in R'(n)$ iff $(w,w') \in R(n)$ and $ (a,a') \in R^{*}(n)$
\item The valuation function $V'$ is such that $V'(w,a) = V(w)$.
\end{itemize}

States at any particular point carry a tag of all preceding actions. The state $(w,a)$ will exist in the final Kripke model only if the precondition of the action $a$ is satisfied by $w$ in the initial model. Two states $(w,a)$ and $(w',a')$ are indistinguishable for Player $x$ in the new Kripke model if and only if the states $w$ and $w'$ were indistinguishable for the Player $x$ in the initial Kripke Model and the actions $a$ and $a'$ were indistinguishable in the action model $\mu$.

 The knowledge of the players about the state of the game at the beginning of the game is modeled by a Kripke model. The knowledge actions which occur during the game are modeled by action models. The knowledge development is modeled by sequential execution of the action model in the Kripke model, resulting in a  new up-to-date Kripke model modeling the knowledge of the players after the knowledge actions. 

\section{Problem Modeling}

 Given the set of players and the basic propositions, the Russian cards problem (RCP) can be modeled in Dynamic Epistemic Logic  (DEL). Let $U = \{ 0, 1, 2, 3, 4, 5, 6 \}$ be the set of cards and  $N = \lbrace a, b, c\rbrace$         (representing Anne, Bill and Cath) be the set of players. The basic propositions are `card 0 is with Anne,' `card 3 is with Bill' and so on. If $i_{x}$  denotes `card $i$ is with $x$' then the set of basic propositions $P = \{ i_{x} \mid x \in N ,\textnormal{  } i \in  U \}$. 

Initially Player $a$ and Player $b$ have three cards each and Player $c$ has one card. The Kripke model for the initial game state is given by $ M = \langle W, R, V\rangle$ 
where,\\
$W = \{(A,B,C) | \textnormal{  }|A|=|B|=3,| C|=1, A\cup B\cup C = U \}$\\
$R(a) = \{ ((A,B,C),(A',B',C'))\mid A = A'\}$\\
$R(b) = \{ ((A,B,C),(A',B',C'))\mid  B = B'\}$\\
$R(c) = \{ ((A,B,C),(A',B',C'))\mid  C = C'\}$\\
$V\big((A,B,C)\big) = \{ i_{a}\mid i\in A\} \cup \{ i_{b}\mid i\in B\} \cup \{ i_{c}\mid i\in C\}$

In any state $w= (A,B,C)$ the set of cards with Player $a$ is $A$, the set of cards  with Player $b$ is $B$ and that with Player $c$ is $C$. $R(x), x \in N$ contains the state pairs that are indistinguishable for Player $x$. The definition of $R(x)$ follows from the fact that initially Player $x$ knows only the cards she is holding. Hence all the states in which her hand of cards is the same will be indistinguishable for her.

For each $A \subseteq U, |A| = 3$, let $T_{A} = \{ (A',B',C'): A' = A$,  $| B'| = 3$, $| C'| = 1 \textnormal{ and } A'\cup B'\cup C' = U\}$.Similarly for each $B \subseteq U, |B| = 3$, let $S_{B} = \{ (A',B',C'): B' = B$, $|A'|  = 3$, $| C'| = 1 \textnormal{ and } A'\cup B'\cup C' = U\}$ and for each $C \subseteq U, |C| = 1$, let $Q_{C} = \{ (A',B',C'): C' = C$, $|A'|  = | B'| = 3 \textnormal{ and } A'\cup B'\cup C' = U\}$.

Hence we have:\\
$R(a) = \bigcup_{A\subseteq U, |A| = 3} T_{A}\times T_{A}$ and  $A\neq A' \Rightarrow T_{A} \cap T_{A'} = \emptyset$\\
$R(b) = \bigcup_{B\subseteq U, |b| = 3} S_{B}\times S_{B}$ and  $B\neq B' \Rightarrow S_{B} \cap S_{B'} = \emptyset$\\
$R(c) = \bigcup_{C\subseteq U, |C| = 1} Q_{C}\times Q_{C}$ and  $C\neq C' \Rightarrow Q_{C} \cap Q_{C'} = \emptyset$\\

\label{components}
 $R(a)$ is a partition of $W$ and $T_{A}$ will be called a {\em component} of $R(a)$. Here $W = \bigcup_{A \subseteq U, |A| = 3} T_{A}$ and the union is disjoint. $R(a)$ has  $\binom{7}{3} = 35$ components  each having $\binom{4}{3} \times \binom{1}{1} = 4$ states. All four states in any component $T_{A}$  are indistinguishable for Player $a$ by the definition of $R(a)$.

 Similarly $R(b)$ is a partition of $W$ and $S_{B}$ is called a component of $R(b)$. $W = \bigcup_{B\subseteq U, |B| = 3} S_{B}$ and the union is disjoint. $R(b)$ has $\binom{7}{3} = 35$ components  each having $\binom{4}{3} \times \binom{1}{1} = 4$ states  and all four states in any component $S_{B}$ are indistinguishable for Player $b$.

 $R(c)$ is a partition of $W$ and $Q_{C}$ is called a component of $R(c)$.  $W = \bigcup_{C \subseteq U, |C| = 1} Q_{C}$ and the union is disjoint. $R(c)$ has  $\binom{7}{1} = 7$ components  each having $\binom{6}{3} \times \binom{3}{3} = 20$ states. The twenty states in a component $Q_{C}$ are indistinguishable for Player $c$.

 As an example for $A = \{0,1,2\} $, $T_{A}=\{(\{0,1,2\}$, $ \{3,4,5\}$, $ \{6\})$, $(\{0,1,2\}$, $\{3,4,6\}$, $\{5\})$,  $(\{0,1,2\} \{3,5,6\}$, $\{4\})$, $(\{0,1,2\}$, $\{4,5,6\}$, $\{3\})\}$. Player $a$ cannot distinguish between these four states because in all the four states Player $a$'s hand is $\{0,1,2\}$. 

Without loss of generality let us assume that Player $a$ is having cards $\{0,1,2\}$, Player $b$ is having $\{3,4,5\}$ and Player $c$ is having $\{6\}$ initially. We denote this state by $w^{*}$. Thus $w^{*}=(\{0,1,2\}$,  $\{3,4,5\}$, $ \{6\})$. Suppose that a single  announcement scheme solves the RCP. In the final model we claim that there will be at least one component $T_{A}$\footnote{Let the initial state be $M = (W,R,V)$ and the action model be $\mu = (A,R^{*},\pi)$. The states in the final model $M \otimes \mu $ will be a subset of $W \times A$. But for notational convenience, we ignore the action-tags in the states. Hence the  set of states in the final model is seen as a subset of $W$.} and $S_{B}$ of the partitions generated by $R(a)$ and $R(b)$ respectively such that $T_{A}= S_{B} = \{w^{*}\}$. This is because if one more state was present in the component,  the Players cannot distinguish between those states. Also there will be at least one component $Q_{C}$ of the partition generated by $R(c)$  such that $w^{*} \in Q_{C} $ and $|Q_{C}| > 1 $. Otherwise Player $c$ will be able to find out the actual state. we present the claim formally:
\begin{lemma}
\label{claim1}
 Assume that RCP is solved in Kripke model $M= (W,R,V)$ and for all $A$, $B$ and $C$ such that $ |A|=|B|=3$ and $ |C|=1$, $T_{A}, S_{B}$ and $Q_{C}$ are components of $R(a), R(b)$ and $R(c)$ respectively. Then the following statements\footnote{These statements are necessary but not sufficient.} hold:
\begin{enumerate}
 \item $\exists A, \exists B$ such that $T_{A}= S_{B} = \{w^{*}\}$\footnote{Actually it should be $T_{A}= S_{B} = \{(w^{*},t)\}$, where $t$ is the list of actions that led to the final state, but for the ease of writing we have omitted the action tag list $t$. }
\item $\forall C $, $w^{*} \in Q_{C} \Rightarrow |Q_{C}| > 1 $.
\end{enumerate}
\end{lemma}
\begin{proof}
Let $w^{*} \in T_{A}$. If possible let $T_{A}$ contain another state --- say $w_{1}$ such that $w^{*} \neq w_{1}$. Let $w^{*} = (A, B^{*}, C^{*})$ and $w^{1}= (A,B^{1}, C^{1})$. $B^{*} \neq B^{1}$ and $C^{*} \neq C^{1}$, otherwise $w^{*} = w^{1} $. Player $a$
cannot distinguish whether Player $b$ is having $B^{*}$ or $B^{1}$. Also she cannot distinguish  whether Player $c$ is having $C^{*}$ or $C^{1}$. This means the RCP is not yet solved contradicting the hypothesis. Therefore $T_{A}=  \{w^{*}\}$. Similar argument shows $S_{B} = \{w^{*}\}$.

$\forall C $, $w^{*} \in Q_{C} \Rightarrow |Q_{C}| > 1 $. If $|Q_{C}| = 1 $, then  $c$ will understand which state she is in.  
\end{proof}

There does not exist even a single card that belongs to Player $a$ in all the states of $  Q_{C}$. If such a card exists, say 4, Player $c$ will understand that 4 is with Player $a$. (Recall the semantics for $K_{n}\phi$ is discussed in section \ref{del} ). Similarly, there does not exist even a single card that belongs to Player $b$ in all the states of $ Q_{C}$.

The above claim is tuned to our requirement of solving the RCP in one announcement. It is easily seen that the claim holds in any final model reached by \textit{any} sequence of announcements.
\section{Lower Bound for RCP}

\begin{theorem} \label{no1RCP}
 There exist no single announcement solution to the RCP.
\end{theorem}

\begin{proof} 
For the sake of contradiction assume there exists a single announcement solution to RCP. Without loss of generality it can be assumed that Player $a$ is making the announcement. Let the action model for the announcement be $\mu=(A,R^{*},\Pi)$.

 In section \ref{del} we have seen that $R(a)$ partitions $W$ into 35 components. Each component will have 4 states . Suppose $w_{1},w_{2},w_{3}$ and $w_{4}$ belong to one component --- say $T_{A}$. Let $w_{1}$ be the actual state. Since Player $a$ will be deterministically making the announcement, she will make the same announcement, say $a_{1}$ for all the elements in $T_{A}$.

In section \ref{action} we have seen that the states $(w_{1},a_{1})$ and $(w_{2},a_{2})$  will be be indistinguishable for Player $a$ if $(w_{1},w_{2})\in R(a)$ and $(a_{1},a_{2})\in R^{*}(a)$. So in the final model $(w_{1},a_{1})$, $(w_{2},a_{1})$,  $(w_{3},a_{1})$ and $(w_{4},a_{1})$ will belong to the same component of $R(a)$. This contradicts Claim \ref{claim1}. Hence there cannot be a single announcement solution to the RCP. 
\end{proof} 
\section{A Generalization}

In this section we will consider a natural generalization of the RCP in which Anne and Bill are holding $k$  cards each and Cath is holding $l$ cards from a pack of $2k+l$ cards. We denote this version of the RCP as RCP($k;l$). Hence the original RCP discussed before is RCP($3;1$) in the new notation.

It can be easily seen that for any $k \geq 1 $ and $l \geq 1$ there does not exist a one announcement solution for RCP($k;l$) as the Theorem \ref{no1RCP} and the proof extends to this case as well . We will now examine the impossibility of a two announcement  solution for RCP($k;l$) using similar strategies.

The set of cards $U = \{ 1, 2, \dots , 2k+l\}$. The initial Kripke model $ M = \langle W, R, V\rangle$ where,\\
$W = \{(A,B,C) | \textnormal{  }|A|=|B|=k,| C|=l, A\cup B\cup C = U \}$\\
$R(a) = \{ ((A,B,C),(A',B',C'))\mid A = A'\}$\\
$R(b) = \{ ((A,B,C),(A',B',C'))\mid  B = B'\}$\\
$R(c) = \{ ((A,B,C),(A',B',C'))\mid  C = C'\}$\\
$V\big((A,B,C)\big) = \{ i_{a}\mid i\in A\} \cup \{ i_{b}\mid i\in B\} \cup \{ i_{c}\mid i\in C\}$

$T_{A},S_{B}$ and $Q_{C}$ are also defined similarly as in Section \ref{components}. It follows that $R(a)$ and $R(b)$ will have $\binom{2k+l}{k}$ components each with $\binom{k+l}{k}\times \binom{l}{l} = \binom{k+l}{k}$ elements each.

Suppose Player $a$ is making an announcement $\alpha$ for a set of components $\mathcal{T}_{\alpha}$. Since Player $c$ should not learn about a single card other than his own hand, we get the following lemma.
\begin{lemma}\label{cond1}
\begin{equation}  
\bigcup_{T_{A}\in \mathcal{T}_{\alpha}}A = U
\end{equation}
\begin{equation}  
\bigcap_{T_{A}\in \mathcal{T}_{\alpha}}A = \emptyset
\end{equation}
\end{lemma}
\begin{proof}
  Assume that this was not the case. i.e.,
$\bigcup_{T_{A}\in \mathcal{T}_{\alpha}}A = Q \mbox{     ,      } Q \subset U$.\\
Player $c$ can infer that the cards in $U\setminus Q$ are not with Player $a$. Hence $U\setminus Q$ has to be $\emptyset$.

Similarly it can be seen that
\label{cond2}
$ \bigcap_{T_{A}\in \mathcal{T}_{\alpha}}A = \emptyset $ since if $\bigcap_{T_{A}\in \mathcal{T}_{\alpha}}A = Q $\mbox{ ,   $   Q\subseteq U $,  } $ Q \neq \emptyset$ , then Player $c$ can infer that the set of cards $Q$ is with Player $a$. 
\end{proof}
Before proving the impossibility of a two announcement  solution for RCP($k;l$) we need to prove the following technical lemma.
\begin{lemma}\label{technical lemma}
 For $k \geq 2$, $l \geq \frac{2k^{2}}{\ln k}$
\begin{equation} \label{techlemmaeqn}
 \lceil \frac{2k+l}{k}\rceil \times  \binom{k+l}{k} > \binom{2k+l}{k}
\end{equation}

\end{lemma}

\begin{proof}
 Enough to have 
\begin{displaymath}
   \frac{2k+l}{k} \times \binom{k+l}{k} >\binom{2k+l}{k}
\end{displaymath}
\begin{displaymath}
 \mbox{i.e.,} \frac{2k+l}{k} \times \frac{(k+l)!}{k! \times l!} > \frac{(2k+l)!}{k! \times (k+l)!}
\end{displaymath}

\[\frac{2k+l}{k} \prod_{i=1}^{k}(l+i) > \prod_{i=1}^{k}(l+k+i)\]

\[ \mbox{i.e.,} \frac{2k+l}{k} > \prod_{i=1}^{k} \left(1 + \frac{k}{l+i}\right).\]

Since \( \left( 1 + x \right) \leq e ^{x}\)
\[ RHS \leq e ^{k \sum_{i=l+1}^{l+k}\frac{1}{i}}.\] 

Bounding the summation by integral we get,
\[RHS \leq e^{k. \ln\left(\frac{l+k}{l} \right)} = \left( \frac{l+k}{l}\right)^{k}.\]

\[\mbox{Enough to have } \frac{2k+l}{k} > \left( \frac{l+k}{l}\right)^{k}. \]

Let \( l = \frac{2k^{2}}{\ln k}\).

\[\mbox{Enough to have } 2 + \frac{2k}{\ln k} > \left( 1+ \frac{\ln k}{2k}\right)^{k} \]
\[ \mbox{since }  \left( 1 + x \right) \leq e ^{x} \mbox{ we get }  \left( 1+ \frac{\ln k}{2k}\right)^{k} \leq e^{\frac{\ln k}{2}}= \sqrt{k} \]
\[ \therefore \mbox{ }2 + \frac{2k}{\ln k} > \sqrt{k} .\]

 We can see that the above equation holds  for all $k \geq 2$. 
\end{proof}

Now we will prove that for sufficiently large $l$, there exist at least two states in $\bigcup_{T_{A}\in \mathcal{T}_{\alpha}}T_{A}$ which are indistinguishable for player $b$ for any announcement $\alpha$ satisfying Lemma \ref{cond1} .

\begin{lemma} \label{2statesnexist}
For $k \geq 2$, $l>\frac{2k^{2}}{\ln k} $ for any announcement $\alpha$ satisfying Lemma \ref{cond1} $ \exists s_{1},s_{2} \in \bigcup_{T_{A}\in \mathcal{T}_{\alpha}}T_{A}$ such that $s_{1} \neq s_{2}$ and $(s_{1},s_{2}) \in R(b)$ 
\end{lemma}

\begin{proof}
 Assume that there do not exist two indistinguishable states for \mbox{ Player $b$} in $\bigcup_{T_{A}\in \mathcal{T}_{\alpha}}T_{A}$. From Lemma \ref{cond1} we will get $|\mathcal{T}_{\alpha}| \geq \lceil \frac{2k+l}{k}\rceil$ since $|A| = k$ and $|U| = 2k+l$ (as the $2k+l$ elements need to be distributed among sets of size $k$). Now the number of different hands possible for Player $b$ should be at least \mbox{$\lceil \frac{2k+l}{k}\rceil \times \binom{k+l}{k}$}. But the different number of combinations possible is $\binom{2k+l}{k}$. By Lemma \ref{technical lemma} $\lceil \frac{2k+l}{k}\rceil \times \binom{k+l}{k}> \binom{2k+l}{k}$ for $l \geq \frac{2k^{2}}{\ln k}$ and $k \geq 2 $. 
\end{proof}

Now we will prove that there does not exist a two-announcement solution for RCP($k;l$) if $k \geq 2$ and $l \geq \frac{2k^{2}}{\ln k}$.
\begin{theorem} 
 For $k \geq 2 $, $l \geq \frac{2k^{2}}{\ln k}$ , there exists no two announcement solution to the RCP($k;l$).
\end{theorem}
\begin{proof}
 The initial Kripke model $M_{1}= \langle W_{1},R_{1},V_{1}\rangle$ as described above. Let Player $a$ make the first announcement $\alpha$. The action model for the first announcement is $\mu_{1}= \langle A_{1},R^{*}_{1},\Pi_{1}\rangle$. $\alpha \in A_{1}$.  Let $s_{1},s_{2} \in W_{1}$ be the two states which are indistinguishable for Player $b$ in $\bigcup_{T_{A}\in \mathcal{T}_{\alpha}}A$. $(s_{1},s_{2}) \in R_{1}(b)$. Let the resulting Kripke model be $M_{2}=\langle W_{2},R_{2},V_{2}\rangle$. The states $(s_{1},\alpha), (s_{2},\alpha) \in W_{2}$ will be indistinguishable for player $b$ since $(s_{1},s_{2}) \in R_{1}(b)$ and $(\alpha,\alpha) \in R^{*}_{1}(b)$. Thus $((s_{1},\alpha), (s_{2},\alpha)) \in R_{2}(b)$. Hence Player $b$ will make the same announcement say $\beta$  for both these states. Let $\mu_{2}= \langle A_{2},R^{*}_{2},\Pi_{2}\rangle$ be the action model for the second announcement with  $\beta \in A_{2}$. Let the resulting Kripke model be $M_{3}=\langle W_{3},R_{3},V_{3}\rangle$. Now the states $(s_{1},\alpha,\beta),(s_{2},\alpha,\beta) \in W_{3}$ will still be indistinguishable for Player $b$ since $((s_{1},\alpha), (s_{2},\alpha)) \in R_{2}(b)$ and $(\beta,\beta) \in R^{*}_{2}(b)$. i.e., $((s_{1},\alpha,\beta), (s_{2},\alpha,\beta)) \in R_{3}(b)$. Hence two announcements are not sufficient to solve the RCP($k;l$) if $l \geq \frac{2k^{2}}{\ln k}$. 
\end{proof}

\section{Conclusion}

We have analyzed the  Russian Cards Problem and the generalization RCP($k;l$) in the framework of Dynamic Epistemic Logic. It is shown  that there does not exist a single announcement solution for the Russian Cards Problem within  the framework of Dynamic Epistemic Logic. Since the framework is considered sufficiently general \cite{dit4,baltag,benthem1lonely,lutz,dit5} we claim that there can be no one-announcement solution to RCP in general and no two announcement solution to  RCP($k;l$) for $l \geq \frac{2k^{2}}{\ln k}$, $k \geq 2$. 

The problem of deriving upper and lower bounds for  RCP($k;l$) in general in terms of $k$ and $l$ remains open for further investigation.

\subsection*{Acknowledgements}
The authors would like to thank Dr. L. Sunil Chandran for helpful
discussions and suggestions. The first author would like to thank Prof. R. Ramanujam for introducing her to the Russian cards problem.

\renewcommand{\baselinestretch}{1.0}
        \begin{small}
         \bibliographystyle{./IEEEtran}% REFERENCES
	 \bibliography{./myreference}
        \end{small}

\end{document}